\documentclass[prl,showpacs,amsmath,amssymb,twocolumn,10pt]{revtex4}
\usepackage[thinlines,thiklines]{easybmat}
\usepackage{amsthm}
\usepackage{enumerate}
\usepackage{color}
\usepackage{dcolumn}
\usepackage{bm}
\usepackage{graphicx}
\newtheorem*{corollary}{Corollary}

\newtheorem{theorem}{Theorem}

\newcommand{\supp}{\mathrm{supp}}
\newcommand{\bra}[1]{\langle #1|}
\newcommand{\ket}[1]{|#1\rangle}

\newcommand{\ip}[2]{\langle #1|#2\rangle}
\newcommand{\op}[2]{|#1\rangle \langle #2|}
\newcommand{\slocc}{\overset{\underset{\mathrm{SLOCC}}{}}{\longrightarrow}}
\newcommand{\Span}{\textrm{span}}

\begin{document}
\date{\today}
\title{\Large {\bf Tripartite to Bipartite Entanglement Transformations and Polynomial Identity Testing}}
\author{Eric Chitambar$^1$}
\email{echitamb@umich.edu}
\author{Runyao Duan$^{2}$}
\email{Runyao.Duan@uts.edu.au}
\author{Yaoyun Shi$^{3}$}
\email{shiyy@eecs.umich.edu} \affiliation{
$^1$Physics Department, University of Michigan, 450 Church Street,
Ann Arbor, Michigan 48109-1040, USA.\\
$^2$Department of Computer Science and Technology, Tsinghua
University, Beijing 100084, China and Center for Quantum Computation
and Intelligent Systems (QCIS), Faculty of Engineering and
Information Technology, University of Technology, Sydney, NSW 2007,
Australia
\\
$^3$Department of Electrical Engineering and Computer Science,
University of Michigan, 2260 Hayward Street, Ann Arbor, MI
48109-2121, USA}

\begin{abstract}
We consider the problem of deciding if a given
three-party entangled pure state can be converted, with a non-zero
success probability, into a given
two-party pure state through local quantum operations and classical communication.
We show that this question is equivalent to the
well-known computational problem of deciding if a multivariate polynomial is identically zero.
Efficient randomized algorithms developed to study the latter can
thus be applied to the question of tripartite to bipartite
entanglement transformations.
\end{abstract}

 \maketitle

A basic question concerning quantum entanglement is whether it can
be transformed in a particular manner.  More specifically, given a
multipartite system originally in some entangled state, what are the
possible final states the system can realize if the only allowed
quantum operations are performed locally, one subsystem at a time
while assisted by global classical communication?  Such protocols
are called \textit{Local Operations and Classical Communication (LOCC)},
and any LOCC transformation
that yields a certain final state with just a nonzero probability is
called \textit{Stochastic LOCC (SLOCC)}.  For some initial state $\ket{\psi}$ and
target state $\ket{\phi}$, an SLOCC conversion between $\ket{\psi}$
and $\ket{\phi}$ is denoted as $\ket{\psi}\slocc\ket{\phi}$.
In this  Letter, we study the computational problem of 
deciding if $\ket{\psi}\slocc\ket{\phi}$ is feasible, given the classical description of $\ket{\psi}$ and $\ket{\phi}$.

In general, deciding the feasibility of state transformations
is difficult, with the challenge becoming more formidable as more
parties are involved.  This increase in difficulty can be made
precise by using the language of computational complexity theory~\cite{Sipser-1997a},
which groups problems according to the amount of resources needed to
solve them. For SLOCC transformations of \textit{bipartite} pure
states, the computational resources in deciding conversion
feasibility increases polynomially in the dimension of either
subsystem undergoing the transformation \cite{Vidal-1999a}; thus the
SLOCC bipartite conversion problem belongs to the complexity class
P.  On the other hand, it was recently shown that the SLOCC
convertibility between \textit{tripartite} states is an NP-Hard
problem, implying that no polynomial time decision algorithm exists if P$\not=$ NP, where NP is a central complexity class
that includes numerous naturally occurring problems not known to be in P.
These results can be
interpreted as a formal indication that it is generally much easier
to study pure state entanglement transformations in bipartite
systems than in tripartite systems.

A natural next step is to examine a hybridization of these two
transformation classes and see where the problem of tripartite to
bipartite SLOCC entanglement conversions fits in the complexity
spectrum.  This is a subset of tripartite transformations in which
Alice, Bob and Charlie initially share a three-way entangled state
$\ket{\psi}_{ABC}$, but they end with a state
$\ket{\phi}_{AB}$ where Alice and Bob are still entangled
but Charlie is completely unentangled from the other two.

In this Letter, we show that deciding tripartite to bipartite
convertibility is equivalent to Polynomial Identity Testing (PIT),
which is the task of determining whether two polynomials given in algebraic formulas are equivalent
(or equivalently, if a polynomial is identically zero).
PIT is a classical problem in theoretical computer science
with many important applications, such as in perfect matching~\cite{Edmonds-1967a},
multiset equality testing~\cite{Blum-1995a}, and primality testing~\cite{Agrawal-2003a}.
In particular, it is known that PIT admits a polynomial time {\em randomized} algorithm but
is not known to have a {\em deterministic} polynomial time algorithm.
The failure to ``derandomize'' the algorithm is shown~\cite{Kabanets-2004a} to arise
from the difficulty of proving super-polynomial lower bounds on general computation models:
if PIT is in P, then some other problems would not have an efficient algorithm.
Thus PIT has been a central problem to the fundamental (and open) question if randomness is useful
in computation (i.e. the BPP versus P question).  The equivalence of the convertibility question to PIT implies
the former also admits a randomized polynomial time algorithm, while it will remain
a difficult question if it has a polynomial time algorithm. Together with
the previously mentioned complexity results on other convertibility questions,
those findings establish close relationships between
quantum entanglement and computational complexity theory.

In the rest of this Letter, we present the formal statement of our result
and its proof. We start with the necessary notation and some observations.
For a tripartite pure state $\ket{\psi}_{ABC}$, Alice and Bob's
joint state is described by
$\rho^\psi_{AB}=Tr_C({}_{ABC}\op{\psi}{\psi}_{ABC})$,
with a mixed state representation $\sum_{i=1}^np_i\op{e_i}{e_i}$,
where $\ip{e_i}{e_j}=\delta_{ij}$ and $p_i>0$.  The ``subnormalized''
eigenstates $\{\ket{\tilde{e}_i}=\sqrt{p_i}\ket{e_i}\}_{i=1\cdots
n}$ span the space $supp(\rho^\psi_{AB})$ called the support of
$\rho^\psi_{AB}$. An ensemble of pure states
$\{\ket{\tilde{q}_j}\}_{j=1\cdots t}$ satisfies 
$\sum_{j=1}^t \ket{\tilde{q}_j}\bra{\tilde{q}_j} = \rho^\psi_{AB}$
if and only if there exists a unitary matrix
$U=[u_{ij}]_{1\le i, j\le t}$, such that
\begin{equation}\label{eqn:density}
\ket{\tilde{q}_j}=\sum_{i=1}^n u_{ij}\ket{\tilde{e_i}},
\end{equation}
where
$\ket{\tilde{e}_i}$ is the zero vector for $i> n$ \cite{Hughston-1993a}.  
Furthermore, there is a one-to-one correspondence
between a measurement of Charlie consisting of rank-one measure operators
and a unitary matrix $U$ that defines the resultant pure state ensemble through Eq.~(\ref{eqn:density})
shared by Alice and Bob following Charlie's measurement.

A consequence of these facts is contained in the
following theorem whose statement requires one final bit of
terminology.  For a bipartite pure state $\ket{\phi}$, its
\textit{Schmidt rank} refers to the minimum number of product states
needed to express it and is denoted by $rk(\ket{\phi})$.
Equivalently, for arbitrary bases $\{\ket{i}_A\}_{i=1\cdots d_A}$
and $\{\ket{i}_B\}_{i=1\cdots d_B}$ of $H_A$ and $H_B$ respectively,
$\ket{\phi}$ can be uniquely identified with a $d_B\times d_A$
matrix $\Phi$ mapping $H_A$ to $H_B$ by $\ket{\phi}=(I\otimes
\Phi)\sum_{i=1}^{d_A}\ket{i}_A\ket{i}_A$, and $rk(\ket{\phi})$
equals the matrix rank of $\Phi$.

\begin{theorem} Let $\ket{\phi}_{AB}$ be a bipartite pure state with Schmidt rank $d$.
Then for tripartite state $\ket{\psi}_{ABC}$, $\ket{\psi}_{ABC}\slocc\ket{\phi}_{AB}$ if and only if
there exists $\ket{\phi'}_{AB}\in \supp(\rho_{AB})$ such that $rk(\ket{\phi'}_{AB})\geq d$.
\end{theorem}
\begin{proof}[\bf Proof]  Let $\ket{\psi}_{ABC}=\sum_{i=1}^{n}\ket{\tilde{e}_i}_{AB}\ket{e_i}_C$
be a Schmidt decomposition of $\ket{\psi}_{ABC}$. ($\Rightarrow$) If
the transformation is possible, then there are linear operators $A$,
$B$, $C$ such that $A\otimes B\otimes
C\ket{\psi}_{ABC}=\ket{\phi}_{AB}\ket{0}_C$ \cite{Dur-2000a}. 
Consequently, $\ket{\phi}_{AB}=A\otimes B\ket{\phi'}$, for 
$\ket{\phi'}=\sum_{i=1}^{n} (C\ket{e_i}_C)\ket{\tilde{e}_i}_{AB}$.
Since $\ket{\phi'} \in \supp(\rho_{AB})$ and $A\otimes B$ cannot increase the Schmidt
rank of $\ket{\phi'}$, we have $rk(\ket{\phi'})\geq d$.
($\Leftarrow$)  Conversely, assume the existence of $\ket{\phi'}\in \supp(\rho_{AB})$ with
$\langle \phi'|\phi'\rangle=1$ and $rk(\ket{\phi'})\geq d$. Then $\ket{\phi'}$
has a unique representation $\ket{\phi'}=\sum_{i=1}^n \alpha_i \ket{\tilde{e}_i}_{AB}$, 
for some complex numbers $\alpha_i$, $1\le i\le n$, with $w=\sum_{i=1}^n |\alpha_i|^2\ge 1$.
Let $\ket{P}= \sum_{i=1}^n \alpha_i \ket{e_i}_C/\sqrt{w}$,
If Charlie applies the projective measurement $\{ \ket{P}\bra{P}, I_C-\ket{P}\bra{P}\}$, with probability
$1/w>0$ he observes $\ket{P}$ and Alice and Bob are left with $\ket{\phi'}$.  Alice and Bob can then convert
$\ket{\phi'}$ into $\ket{\phi}$ with nonzero probability because the
target state's Schmidt rank is not higher~\cite{Vidal-1999a}.
\end{proof}

We note that the above result generalizes the already established
condition of SLOCC convertibility between bipartite pure states of
ref. \cite{Vidal-1999a}.  If the initial joint state of Alice and
Bob is pure then the transformation becomes
$\ket{\psi}_{AB}\ket{0}_C\slocc\ket{\phi}_{AB}$.  In this
case, $\ket{\psi}_{AB}$ is the only state in $supp(\rho_{AB})$ and
the transformation is possible if and only if $rk(\ket{\psi}_{AB})\geq
rk(\ket{\phi}_{AB})$.

A unidirectional protocol like that described above is often called
``one-shot'' as Charlie's involvement consists of just making a
measurement with rank-one measure operators and broadcasting the
result to Alice and Bob. Hence $\ket{\psi}_{ABC}$ can be converted
to $\ket{\phi}_{AB}$ with a nonzero probability if and only
if it can be done so by a one-shot protocol.  The situation is
strikingly different in the case of deterministic transformations
since there exist tripartite to bipartite conversions that require
bidirectional collaboration between the parties in order to occur
with probability one \cite{Gour-2006a}.

According to Theorem 1, the problem of deciding conversion is
reduced to whether a Schmidt rank $d$ state exists in some subspace
of $H_A\otimes H_B$.  This question is a generalization of one
sometimes referred to as Edmonds' Problem \cite{Edmonds-1967a}: if
$M(d_A,d_B)$ is the linear space of $d_B\times d_A$ matrices with
complex coefficients and $V$ is some subspace of $M(d_A,d_B)$,
decide whether there exists a rank $d=\min\{d_A,d_B\}$ (i.e. full
rank) matrix in $V$. 
Intuitively, one would expect that an upper
bound exists on the dimension $s$ of subspaces containing only
states of rank strictly less than $d$.  Indeed, Flanders provides
the bound $s<d\cdot \max\{d_A,d_B\}$
\cite{Flanders-1962a,Cubitt-2008a}.  Thus,
\begin{corollary}
If $rk(\ket{\phi})=d$ and $\dim[supp(\rho^\psi_{AB})]\geq d\cdot
\max\{d_A,d_B\}$, then
$\ket{\psi}_{ABC}\slocc\ket{\phi}_{AB}$.
\end{corollary}

It requires more work when $\dim[supp(\rho^\psi_{AB})]< d\cdot \max\{d_A,d_B\}$.
 As recognized by previous investigators, determining whether a
matrix subspace is singular can be cast into a polynomial identity
testing question \cite{Gurvits-2003a,Lovasz-1989a}.  We generalize
their approach to the subject at hand.  Letting
$\{\tilde{\Pi}_i\}_{i=1 \cdots n}$ denote the nonzero subnormalized
eigenstates of $\rho^\psi_{AB}$ in matrix form, any state in
$supp(\rho^\psi_{AB})$ can be expressed as
\begin{equation}\label{0}
\Pi(\textbf{u})=u_{1}\tilde{\Pi}_1+\cdots+u_{n}\tilde{\Pi}_n,
\end{equation}
where $\textbf{u}=(u_1,\cdots,u_n)$ is an $n$-dimensional complex
vector. Construct the following real-valued function
\begin{equation}
\label{1} g(\textbf{u})=\sum_{\kappa_d} |\det(\kappa_d)|^2,
\end{equation} where $\kappa_d$ ranges over the set of $d\times d$ sub-matrices of
$\Pi(\textbf{u})$ and $\det(\kappa_d)$ denotes the determinant
of $\kappa_d$. Note that $g$ is a nonnegative polynomial of degree no
greater than $2d$ in the real variables $\{a_i,b_i\}_{i=1\cdots n}$
where $u_i=a_i+{\bf i}b_i$. Then deciding whether $g$ is identically
zero is the same as determining whether a Schmidt rank $d$ state is
obtainable from $\ket{\psi}_{ABC}$ since $g(\textbf{u})\equiv 0$ if and only if 
no matrix of rank at least $d$ exists in the span of
$\{\tilde{\Pi}_i\}_{i=1\cdots n}$.

As mentioned above, Polynomial Identity Testing is a classic problem in theoretical
computer science with many important applications
\cite{Motwani-1995a,Lovasz-1989a, Kabanets-2004a}.  
In general, given two polynomials $f(\textbf{x})$ and
$p(\textbf{x})$, it can always be decided if $f=p$ by multiplying
out the polynomials and checking whether their coefficients match.
However, the number of multiplications required for this procedure
scales exponentially in the degree of the polynomials and at the
present no sub-exponential deterministic algorithm is known for polynomial
identity testing \cite{Motwani-1995a, Arora-2009a}.

On the other hand, if one relaxes the deterministic condition,
randomized polynomial-time algorithms exist that can decide with a
high probability of success \cite{Schwartz-1980a, Arora-2009a}.  A
standard algorithm uses the Schwartz-Zippel lemma which states that
for some $n$-variate polynomial $f(x_1,\cdots,x_n)$ over a field $\mathbb{K}$
and having degree no greater than $d$, if $f$ is not identically
zero, then $\text{Prob}[f(x'_1,\cdots,x'_n)=0]\leq\frac{d}{|X|}$
where each $x'_i$ is independently sampled from some finite set
$X\subset \mathbb{K}$. Thus, to test with success probability at least
$1-\frac{d}{|X|}$ whether a polynomial $f$ is identically zero, one
evaluates $f$ on values chosen from set $X$ and decides a zero
identity if and only if $f(x'_1,\cdots,x'_n)=0$.  Using this algorithm on the
analysis of $g$ in Eq. (\ref{1}) allows for the following
classification of the tripartite to bipartite conversion problem.

\begin{theorem}
\label{thm2} There exists a polynomial time randomized algorithm that,
given states $\ket{\psi}_{ABC}$ and $\ket{\phi}_{AB}$, decides correctly
if $\ket{\psi}_{ABC}\slocc\ket{\phi}_{AB}$ is feasible with
probability $\ge 2/3$.
\end{theorem}

In other words, the problem of deciding tripartite to bipartite SLOCC
convertibility belongs to the complexity class BPP, which
consists of decision problems solvable by Bounded-error Probabilistic Polynomial time
algorithms. The error probability $1/3$ can be made exponentially small by repeating
the algorithm and outputting the majority of the outputs of each repetition.

The randomized algorithm described above can be used to construct an LOCC
protocol that completes any feasible tripartite to bipartite
transformation with nonzero probability.  Charlie makes $2n$
independent samplings $a_1,b_1,a_2,b_2,\cdots,a_n,b_n$ from the
integer set $\{1,\cdots,M\}$, where $M$ is an integer larger than
$2d$. Then he constructs
$\textbf{u}=(a_1+\textbf{i}b_1,\cdots,a_n+\textbf{i}b_n)$ and
evaluates $g(\textbf{u})$ in Eq. (\ref{1}). If $g(\textbf{u})$ is
nonzero, the state $\ket{P}=\frac{1}{\sqrt{N}}\sum_{i=1}^n
u_i\ket{e_i}$ is formed with $N$ being the appropriate normalization factor. By the
Schwartz-Zippel lemma, such a $\textbf{u}$ can be found with success
probability at least $1-\frac{2d}{M}$ in the case that $g$ is not
identically zero, and an appropriate $M$ can be chosen to make this
probability sufficiently large. Then Charlie performs a
projective measurement $\{\op{P}{P},I_C-\op{P}{P}\}$, and Alice and Bob
will share the unnormalized state
$\Pi_{\phi'}(\textbf{u})=\sum_{i=1}^n u_i\tilde{\Pi}_i$ with nonzero
probability when the outcome is $\ket{P}$. They will then be
able to probabilistically obtain the desired target state
$\ket{\phi}$ as $\ket{\phi'}$ is with Schmidt rank at least $d$. One
drawback of the above procedure is that we need to evaluate $g(\textbf{u})$.
However, the explicit form of $g$ is unknown and may be very
complicated as we need to sum over all determinants of $d\times d$
sub-matrices of $\Pi(\textbf{u})$. Fortunately, we can avoid
evaluating $g(\textbf{u})$ directly by checking whether the matrix
rank of $\Pi(\textbf{u})$ is larger than $d$, which can be done
efficiently in polynomial time of $n$ and $d_{A(B)}$.

It is interesting how the problem can be turned around by reducing
any polynomial identity testing question to a decision of SLOCC
convertibility.  In an important work by Valiant~\cite{Valiant-1979a} (slightly
improved in~\cite{Liu-2006a}), he shows that any
polynomial $p(x_1,\cdots,x_m)$ over a field $\mathbb{K}$ of formula size $e$ can be expressed as the
determinant of some $(e+2)\times(e+2)$ matrix $\Pi_p(x_1,\cdots,x_m)$
with entries in $\{x_1, ..., x_m\}$ and the underlying field
\cite{Valiant-1979a}, and the construction of $\Pi_p$ from $p$
can be done in polynomial time.
Let $\Pi_p=\Pi_0+\sum_{i=1}^m x_i \Pi_i$, where $\Pi_i\in\mathbb{K}^{(e+2)\times (e+2)}$,
$0\le i\le m$. We claim that $\Span\{\Pi_i : 0\le i \le m\}$ contains a nonsingular matrix
if and only if $\Pi_0+\Span\{\Pi_i : 1\le i\le m\}$ does.
The ``only if'' direction is trivial. Assume that $A=\sum_{i=0}^m \alpha_i \Pi_i$
is nonsingular. If $\alpha_0\ne 0$, then $\Pi_0+\sum_{i=1}^m\frac{\alpha_i}{\alpha_0}\Pi_i$
is nonsingular. Otherwise, for sufficiently large $k$, $\det(\Pi_0 + k A)>0$. Thus
$\Pi_0+\Span\{\Pi_i: 1\le i\le m\}$ contains a nonsingular matrix.

To conclude the reduction, for each $i$, $0\le i \le m$, let $\ket{e_i}$ be the bipartite state
corresponding to $\Pi_i$, and $\rho_{AB}=\frac{1}{n}\sum_{i=1}^n\op{e_i}{e_i}$.
Letting $\ket{\psi}_{ABC}$ be a purification of $\rho_{AB}$ and
$\ket{\phi}_{AB}$ any rank $n$ bipartite state, $p(x_1,\cdots,x_m)$
is not identically zero if and only if $\ket{\psi}_{ABC}\slocc\ket{\phi}_{AB}$.

In conclusion, the results of this  Letter help contribute to the
complexity hierarchy of SLOCC pure state transformations.  For
bipartite transformations, the question of convertibility reduces to
matrix rank calculations, which can be done in deterministic polynomial time. 
As seen here, for tripartite to bipartite conversions, determining feasibility is equivalent to testing the identity
of a given polynomial, which can be done in randomized polynomial time, and
whether or not the algorithm can be derandomized is a major open problem in theoretical
computer science. For general tripartite transformations,
no polynomial time algorithm (deterministic or randomized) is known, and 
under the common belief in complexity theory, no polynomial time algorithm exists.

The next question might be how the
complexity of calculating optimal conversion probabilities compares
for transformations involving a different number of parties.  For
bipartite pure state conversions, the problem is already known to
have a polynomial-time solution \cite{Vidal-1999a}.  It is
interesting how in all these cases, abstract computational questions
can be used to solve a seemingly unrelated physical problem of
entanglement conversion and vice versa.  Such results demonstrate
the value of complexity theory in the study of quantum systems and
reveals an important relationship between the two fields.

E. C. thanks Richard Low, Toby Cubitt and Nengkun Yu for many helpful discussions during the preparation of this Letter.  This work was supported by the National Science Foundation of the
United States under Awards~0347078 and 0622033.  R. Duan was also
partially supported by the National Natural Science Foundation of
China (Grant Nos. 60702080, 60736011).

\bibliography{QuantumBib}

\end{document}